%% file: main.tex
\documentclass[12pt]{article}

\usepackage{amsmath, amsthm, amssymb}
\usepackage{natbib}
\usepackage{graphicx}
\usepackage{booktabs}
\usepackage[letterpaper,margin=1in]{geometry}
\usepackage{hyperref}
\usepackage{setspace}

\def\spacingset#1{\renewcommand{\baselinestretch}{#1}\small\normalsize}
\providecommand{\anon}{1}
\newcommand{\selfrepo}{%
  \if1\anon
    the \texttt{shapCopula} repository
    (\url{https://github.com/agostinognasso/shapCopula}), archived at
    \url{https://doi.org/10.5281/zenodo.22789499}
    \citep{Gnasso2026shapCopula}%
  \else
    a public code repository, archived under a persistent DOI (citation and
    links withheld for double-anonymous review)%
  \fi}

\newtheorem{theorem}{Theorem}

\theoremstyle{definition}
\newtheorem{assumption}{Assumption}

\theoremstyle{remark}
\newtheorem{remark}{Remark}

\newcommand{\E}{\mathbb{E}}
\newcommand{\PP}{\mathbb{P}}
\newcommand{\R}{\mathbb{R}}
\newcommand{\KL}{\mathrm{KL}}
\newcommand{\TV}{d_{\mathrm{TV}}}

\DeclareMathOperator*{\argmin}{arg\,min}

\date{}

\begin{document}

\spacingset{1}

\if1\anon
{
  \title{\bf Semiparametric Inference for Conditional Shapley Feature Importance}
    \author{
  	Agostino Gnasso\thanks{%
  		Address for correspondence: Department of Economics and Statistics,
  		University of Naples ``Federico II'', Via Cintia, Monte S.\ Angelo,
  		80126 Naples, Italy. E-mail: \texttt{agostino.gnasso@unina.it}.
  		}\\
  	Department of Economics and Statistics\\
  	University of Naples ``Federico II'', Naples, Italy}
  \maketitle
} \fi

\if0\anon
{
  \bigskip
  \bigskip
  \bigskip
  \begin{center}
    {\LARGE\bf Semiparametric Inference for Conditional Shapley Feature Importance}
  \end{center}
  \medskip
} \fi

\bigskip
\begin{abstract}
\noindent
Shapley values are widely used for post-hoc feature attribution, but most
estimators return point quantities and do not quantify uncertainty, and
popular implementations sample out-of-coalition features from their marginal
distribution, which misattributes importance when features are dependent.
This paper studies the conditional formulation, in which out-of-coalition
features are integrated out under their true conditional distribution. The
target is a global, loss-based importance that pairs a conditional value
function with a SAGE-style loss aggregation. We propose a one-step estimator
with $K$-fold cross-fitting and a U-statistic correction of the squared loss
that removes the Monte Carlo bias of the naive plug-in; it is
$\sqrt{n}$-consistent and asymptotically normal under double-robust rate
conditions, and the resulting Wald interval attains nominal coverage. A
Pinsker-type bound quantifies the bias from misspecifying the working copula
class, while vine copulas keep conditional sampling tractable. In a
Gaussian-design study with $n=500$, the empirical coverage of the $95\%$
interval lies between $0.91$ and $0.96$ across all features, the test holds
its Type-I rate at $0.05$, and it reaches power one for moderate signals.
Applied to the UCI Concrete and California Housing data, the method
identifies the conditionally informative features with Bonferroni-controlled
significance.
\end{abstract}

\noindent%
\vfill

\newpage
\spacingset{1.8} 

\section{Introduction}

Shapley values~\citep{Shapley1953}, popularized for machine learning by
\citet{LundbergLee2017}, have become a default tool for
post-hoc feature attribution. Widely used implementations such as
TreeSHAP~\citep{Lundberg2020} expose both a path-dependent and an
interventional route, and their recent releases default to the
\emph{marginal} (interventional) value function, which samples
out-of-coalition features from their marginal distribution. When features are dependent, this choice is known to evaluate
the prediction function on out-of-distribution points and inflate the
apparent importance of correlated variables~\citep{HookerMentchZhou2021,
Janzing2020, SundararajanNajmi2020}. The alternative \emph{conditional}
value function, advocated by~\citet{AasJullumLoland2021} and
\citet{Frye2020}, respects the feature dependence structure but requires
estimating conditional distributions over arbitrary coalitions.

The choice between the two is not a settled matter, and it would be wrong to
present the conditional formulation as the only defensible one. The marginal
value function has a principled causal reading: it answers an interventional
question about the prediction function \citep{Janzing2020}, and on that basis
it has been adopted in much of the recent methodological work, including the
inferential treatment of \citet{Whitehouse2026}. Our position is that the two
target different estimands rather than that one is uniformly preferable. When
the question is which features carry information about the prediction under
the law that generated the data---as opposed to the effect of intervening on
them---the conditional formulation is the relevant one, and the extrapolation
induced by marginal sampling is then a genuine cost rather than a modeling
convenience \citep{HookerMentchZhou2021}. Section~\ref{sec:sim} makes that
cost concrete.

Existing work on the conditional formulation is almost entirely focused on
point estimation. \citet{AasJullumLoland2021} propose conditional sampling
using Gaussian, empirical, and copula-based methods, but provide no
uncertainty quantification. On the inferential side, the framework of
\citet{WilliamsonFeng2020, Williamson2023} gives asymptotically valid
confidence intervals for a population-level variable-importance functional,
but under an \emph{independence-invariance} sampling scheme closely tied to
the marginal (not conditional) interpretation. Closest to the present work,
\citet{Whitehouse2026} construct debiased, asymptotically normal estimators
for global summaries of SHAP attributions, combining U-statistics with
Neyman-orthogonal scores for functionals of nested regressions; their target
is a moment of the local attributions and their value function is the marginal
one, so the nuisance parameter they must control is a regression function.
What remains open is the conditional case. To the best of our knowledge no
method pairs a dependence-aware conditional value function with valid
large-sample inference, and the obstruction is precisely that the nuisance is
then a conditional \emph{law} over arbitrary coalitions rather than a single
regression, so that neither the rate conditions nor the orthogonality argument
transfer unchanged.

The same gap is visible in software. The \texttt{shapr} package
\citep{Jullum2025shapr} implements a broad range of conditional value-function
estimators, and \texttt{xplainfi} \citep{Burk2026xplainfi} implements both
marginal and conditional SAGE; neither attaches confidence intervals with
established coverage guarantees to the conditional Shapley importances it
returns. \citet{Burk2026xplainfi} are explicit about why: they
decline to expose the permutation-based intervals of
\citet{CovertLundbergLee2020} as inferential objects, on the grounds that
their coverage properties are not established, and observe that the required
techniques and coverage studies do not yet exist. This paper supplies them for
the conditional formulation.

The contribution of the paper is fourfold. We define a conditional
SAGE parameter $\Psi_j(P)$ that pairs the conditional value function
of \citet{AasJullumLoland2021} with the loss-based aggregation
of \citet{CovertLundbergLee2020}; the parameter is a global,
non-negative scalar, summarizing the contribution of feature $j$ to the
predictive variance of $f$ when other features are marginalized
conditionally. We then construct a one-step estimator with $K$-fold
cross-fitting in the spirit of \citet{Chernozhukov2018} and
\citet{Williamson2023}. A further ingredient is a U-statistic-style
estimator of the squared loss based on two independent Monte Carlo
sub-batches; this removes the $\mathrm{Var}_{X_{-S}\mid X_S}(f)/B$ bias of
the naive plug-in and is what restores nominal coverage at moderate $B$.
Corrections of this type have been developed concurrently and independently by
\citet{Whitehouse2026} for marginal SHAP functionals, and we make no claim of
priority for the device itself; what is specific to our setting is that the
two sub-batches are draws from an \emph{estimated conditional law}, so the
correction has to be carried through the copula nuisance rather than through a
fitted regression. On the
theoretical side we prove asymptotic normality (Theorem~\ref{thm:an})
and Wald coverage (Theorem~\ref{thm:cov}); because $\Psi_j$ is a mean
functional, its canonical gradient is the mean-centered pointwise
Shapley, with no need for kernel localization or density-weighted Riesz
representers. Finally we give a Pinsker-type bound
(Theorem~\ref{thm:mis}) of order $\sqrt{\KL(C^\star\,\|\,\hat C_n)}$,
where $\KL(\cdot\,\|\,\cdot)$ is the Kullback--Leibler divergence, $C^\star$
the true copula and $\hat C_n$ its working-class estimate;
this controls the bias caused by misspecification of the working copula
class; the bound vanishes for nonparametric vines under the simplifying
assumption \citep{NaglerCzado2016}.

Section~\ref{sec:bg} reviews Shapley values, copulas and
influence functions. Section~\ref{sec:meth} defines our estimator.
Section~\ref{sec:theory} states the three main theorems with proof sketches;
full proofs are deferred to the supplementary material. Sections~\ref{sec:sim}
and~\ref{sec:app} report simulations and applications.
Section~\ref{sec:disc} discusses limitations.

\section{Background}\label{sec:bg}

Let $X = (X_1,\dots,X_p) \in \mathcal{X} \subseteq \R^p$ be a continuous
random vector with joint law $P_X$, Lebesgue density $p_X$, marginals
$F_1,\dots,F_p$, and unique copula $C$~\citep{Sklar1959}. Let $Y \in \R$
with $\E|Y| < \infty$ and $f : \mathcal{X} \to \R$ a fixed prediction
function (or a data-driven estimate $\hat f_n$). For $S \subseteq [p] :=
\{1,\dots,p\}$ denote $X_S = (X_j)_{j \in S}$ and $X_{-S} = X_{[p]\setminus
S}$. The \emph{conditional value function} is
\begin{equation}\label{eq:v}
v(S; x) \;=\; \E\!\left[ f(X) \,\middle|\, X_S = x_S \right]
\;=\; \int f(x_S, x_{-S}) \, p_{X_{-S}\mid X_S}(x_{-S}\mid x_S) \, dx_{-S},
\end{equation}
with $v(\emptyset; x) = \E[f(X)]$ and $v([p]; x) = f(x)$. For brevity we
also write $v_S(x) := v(S; x)$.

\subsection{Conditional Shapley values and the SAGE parameter}

For $j \in [p]$ and $x \in \mathcal{X}$, the local conditional Shapley value
is
\begin{equation}\label{eq:shap}
\varphi_j(x) \;=\; \sum_{S \subseteq [p]\setminus\{j\}}
\frac{|S|!\,(p-|S|-1)!}{p!}\,\bigl[\,v(S\cup\{j\};x) - v(S;x)\,\bigr].
\end{equation}
Here $|S|$ is the cardinality of the coalition $S$; writing
$w_S := |S|!\,(p-|S|-1)!/p!$ for the Shapley weights (which sum to one over
$S \subseteq [p]\setminus\{j\}$), $\varphi_j$ satisfies local accuracy
$\sum_j \varphi_j(x) = f(x) - \E f(X)$,
missingness, and consistency~\citep{LundbergLee2017}.
Following~\citet{CovertLundbergLee2020} we move from the local pointwise
attribution to a global, loss-based summary. Define the
\emph{excess-loss value function}
\begin{equation}\label{eq:V-loss}
V(S; P) \;=\; -\,\E_X\!\bigl[(f(X) - v(S; X))^2\bigr],
\end{equation}
where $P = P_X$ is the law of the feature vector and the expectation $\E_X$
is taken over an independent draw $X \sim P_X$ (the conditional means $v(S;X)$
inside the square depend on the same $P$). This value function
is non-positive, vanishes at $S = [p]$, and equals $-\mathrm{Var}(f(X))$
at $S = \emptyset$. The \emph{conditional SAGE} importance is the Shapley
value of $V$:
\begin{equation}\label{eq:psi-global}
\Psi_j(P) \;=\; \sum_{S \subseteq [p]\setminus\{j\}}
\frac{|S|!\,(p-|S|-1)!}{p!}\,\bigl[\,V(S\cup\{j\}; P) - V(S; P)\,\bigr].
\end{equation}
$\Psi_j \geq 0$ with equality if and only if $v(S\cup\{j\}; X) = v(S; X)$
$P_X$-a.s.\ for every $S$, i.e.\ $j$ contributes nothing to the conditional
predictive mean at any coalition. Equation~\eqref{eq:psi-global}
admits a pointwise representation
$\Psi_j(P) = \E_X[\Psi_j^{\mathrm{loc}}(X; P)]$ with
\begin{equation}\label{eq:psi-local}
\Psi_j^{\mathrm{loc}}(x; P) = \sum_{S \subseteq [p]\setminus\{j\}}
\frac{|S|!\,(p-|S|-1)!}{p!}\,\bigl[(f(x)-v_S(x))^2 - (f(x)-v_{S\cup\{j\}}(x))^2\bigr],
\end{equation}
which is exploited in the inferential analysis of Section~\ref{sec:meth}.

\subsection{Copulas, vines, and influence-function inference}

By Sklar's theorem~\citep{Sklar1959}, the joint density decomposes as
$p_X(x) = c(F_1(x_1),\dots,F_p(x_p))\prod_{j=1}^p f_j(x_j)$,
where $c$ is the copula density and $f_j = F_j'$ is the density of the $j$-th
marginal $F_j$ (the subscripted $f_j$ is a marginal density, not to be
confused with the prediction function $f$). The conditional density appearing
in~\eqref{eq:v} is then
\begin{equation}\label{eq:cond-copula}
p_{X_{-S}\mid X_S}(x_{-S}\mid x_S) = \frac{c(F_1(x_1),\dots,F_p(x_p))}
{c_S(\{F_j(x_j)\}_{j\in S})} \prod_{k \notin S} f_k(x_k).
\end{equation}
Here $c_S$ denotes the copula density of the sub-vector $X_S$ (the marginal
copula obtained from $c$ by integrating out the coordinates not in $S$), and
$f_k$ is the marginal density of $X_k$ as above.
For $p$ large, direct estimation of $c$ is infeasible. \emph{Vine copulas}
\citep{Joe2014, Czado2019} decompose $c$ into a cascade of bivariate
conditional copulas along a regular graph (R-vine); each pair-copula can be
estimated nonparametrically with transformed-kernel densities
\citep{Geenens2014, NaglerSchellhaseCzado2017}. Conditional sampling proceeds by
sequential inversion of the associated $h$-functions at cost $O(p^2)$ per
draw.

Let $\theta : \mathcal{P} \to \R$ be a pathwise-differentiable functional
on a family $\mathcal{P}$ of candidate laws,
with canonical gradient (efficient influence function) $\psi_\theta$
satisfying $\E_P \psi_\theta(X) = 0$
and
\[
\theta(P_1) - \theta(P_0) \;=\; \int \psi_\theta \, d(P_1 - P_0) + R(P_1,P_0),
\]
with $R$ a second-order remainder~\citep{BKRW1993, vanderVaart2000}. Given a
nuisance estimator $\hat P_n$, the one-step estimator is $\theta(\hat P_n) +
n^{-1}\sum_i \hat\psi_\theta(X_i)$. Cross-fitting~\citep{Chernozhukov2018}
with $K$ folds eliminates the overfitting bias and, under double-robust
product rates, gives $\sqrt{n}$-asymptotic normality.
\citet{WilliamsonFeng2020, Williamson2023} instantiate this program for
variable-importance parameters under an independence-style sampling scheme;
our contribution is to carry out the analogous construction for the
conditional Shapley functional, where the nuisance is the full joint copula.

\section{Methods}\label{sec:meth}

\subsection{Influence function and debiased squared loss}

The parameter $\Psi_j(P)$ is a mean functional in $X$, and at a fixed
nuisance $v$ its canonical gradient is simply
\begin{equation}\label{eq:if-psi}
\psi_{\Psi_j}(x; P) \;=\; \Psi_j^{\mathrm{loc}}(x; P) - \Psi_j(P),
\end{equation}
which has mean zero under $P$ and variance equal to that of
$\Psi_j^{\mathrm{loc}}(X)$. This is much simpler than the canonical gradient
of the local parameter $\varphi_j(x_0)$, which would require either a
density-weighted Riesz representer or a kernel localization around $x_0$;
the simplification is what allows the Wald interval to attain nominal
coverage at moderate sample sizes (Section~\ref{sec:sim}).

A natural plug-in for $\delta_S(x) := (f(x) - v_S(x))^2$ would draw a
single batch of $B$ conditional samples, average them into $\hat v_S(x)$,
and return $(f(x) - \hat v_S(x))^2$. This estimator is biased upward by
$\mathrm{Var}_{X_{-S}\mid X_S}(f(X))/B$, and at moderate $B$ the bias is
not negligible; in our experiments it is the single source of finite-sample
under-coverage we have observed. We avoid it by splitting the $B$ draws
into two independent sub-batches of size $B/2$, computing a separate
average $\hat v_S^{(a)}(x)$ and $\hat v_S^{(b)}(x)$ from each, and forming
\begin{equation}\label{eq:debiased}
\widehat\delta_S^{\mathrm{unb}}(x) \;=\;
\bigl(f(x) - \hat v_S^{(a)}(x)\bigr)\bigl(f(x) - \hat v_S^{(b)}(x)\bigr).
\end{equation}
Because the two sub-batches are independent given $x$,
$\E[\widehat\delta_S^{\mathrm{unb}}(x) \mid x] = \delta_S(x)$
exactly. The construction can be seen as a degree-two U-statistic over the
$B$ Monte Carlo draws.

\subsection{One-step cross-fitted estimator}

The data are split into $K$ folds with index sets
$\mathcal{I}_1,\dots,\mathcal{I}_K$. On the complement of fold $k$ a
nuisance copula $\hat C_n^{-k}$ is fitted, and for every observation
$i \in \mathcal{I}_k$ the pointwise Shapley estimate is
\begin{equation}\label{eq:psi-loc-hat}
\widehat{\Psi_j^{\mathrm{loc}}}(X_i) \;=\;
\frac{1}{M}\sum_{m=1}^M \bigl[\widehat\delta_{S^{(m)}}^{\mathrm{unb}}(X_i)
- \widehat\delta_{S^{(m)} \cup \{j\}}^{\mathrm{unb}}(X_i)\bigr],
\end{equation}
where $M$ is the number of sampled permutations, $\pi^{(m)}$ are uniformly
random permutations of $[p]$, and
$S^{(m)} = \pi^{(m)}_{<j}$ collects the features that precede $j$ along
the $m$-th permutation. The one-step estimator and its standard error
take the simple sample-mean form
\begin{equation}\label{eq:os-sage}
\hat\Psi_j \;=\; \frac{1}{n}\sum_{i=1}^n \widehat{\Psi_j^{\mathrm{loc}}}(X_i),
\qquad
\widehat{\mathrm{SE}}_j \;=\; \mathrm{sd}\bigl(\widehat{\Psi_j^{\mathrm{loc}}}\bigr)/\sqrt{n},
\end{equation}
and the Wald confidence interval is
$\hat\Psi_j \pm z_{1-\alpha/2}\,\widehat{\mathrm{SE}}_j$, where $\alpha$ is
the nominal level and $z_{1-\alpha/2}$ the $(1-\alpha/2)$-quantile of the
standard normal distribution. No separate
plug-in correction term is needed: under cross-fitting the nuisance bias
is of second order and is dominated by the central limit term
(Section~\ref{sec:theory}). Throughout the experiments we use $K=3$,
$M=30$, and $B=20$ unless noted otherwise; sensitivity to these choices
is reported in Section~S3 of the supplementary material. Unless stated otherwise the
nuisance $\hat C_n^{-k}$ is the Gaussian-copula working class (latent
Gaussian dependence on the normal-score scale, empirical marginals; the
standard sampler of \citealp{AasJullumLoland2021}), fitted on each training
fold. This class is
correctly specified for the Gaussian designs of Sections~\ref{sec:sim}
(Sim~A--C), so Assumption~\ref{ass:C} holds there; Sim~D and the
applications probe the misspecified regime governed by
Theorem~\ref{thm:mis}. The nonparametric vine of \citet{NaglerCzado2016}
is the more flexible alternative under which Assumption~\ref{ass:C} holds
without a parametric restriction.

A reference implementation of the proposed estimator, together with all
simulation and application scripts, is provided in \selfrepo.

\section{Theory}\label{sec:theory}

We state the main results under the following conditions.

\begin{assumption}[Density]\label{ass:dens}
$X$ has a continuous density, bounded below on a compact support
$\mathcal{X}$, with absolutely continuous marginals.
\end{assumption}

\begin{assumption}[Outcome model]\label{ass:f}
$f \in L^2(P_X)$ and $\|\hat f_n - f\|_{L^2(P_X)} = o_P(n^{-1/4})$.
\end{assumption}

\begin{assumption}[Copula class]\label{ass:C}
The nonparametric vine estimator satisfies
$\|\hat C_n - C\|_{L^2} = o_P(n^{-1/4})$. We impose this rate as a
condition on the sampler. \citet{NaglerCzado2016} establish consistency and
convergence rates for the \emph{density} of a simplified nonparametric vine;
passing from those rates to an $L^2$ rate on the copula itself requires an
additional integration step that we do not carry out, so Assumption~\ref{ass:C}
should not be read as a corollary of that work.
\end{assumption}

\begin{assumption}[Finite IF variance]\label{ass:psi}
The pointwise Shapley $\Psi_j^{\mathrm{loc}}(\cdot;P)$ has finite, strictly
positive variance under $P$, $\sigma_j^2 = \mathrm{Var}_P(\Psi_j^{\mathrm{loc}}(X)) \in (0,\infty)$,
and the map $v \mapsto \Psi_j^{\mathrm{loc}}(\cdot;v)$ is Lipschitz in a
neighborhood of the truth.
\end{assumption}

\subsection{Asymptotic normality and coverage}

\begin{theorem}[Asymptotic normality]\label{thm:an}
Under Assumptions~\ref{ass:dens}--\ref{ass:psi}, with $K$-fold cross-fitting
$(K \geq 2)$,
\[
\sqrt n\bigl(\hat\Psi_j - \Psi_j\bigr)
\;\xrightarrow{d}\; \mathcal{N}\!\bigl(0,\sigma_j^2\bigr),
\qquad \sigma_j^2 = \mathrm{Var}_P(\Psi_j^{\mathrm{loc}}(X)),
\qquad \hat\sigma_j^2 \xrightarrow{P} \sigma_j^2.
\]
\end{theorem}

\begin{proof}[Sketch]
Decompose
$\hat\Psi_j - \Psi_j = n^{-1}\sum_i \bigl[\widehat{\Psi_j^{\mathrm{loc}}}(X_i) - \Psi_j\bigr]$.
The nuisance error term is $o_P(n^{-1/2})$ by cross-fitting and
Assumption~\ref{ass:C}: independence of the two debiased sub-batches
in~\eqref{eq:debiased} gives unbiasedness conditional on the nuisance,
and the second-order remainder is bounded by
$C \cdot \max_S \|\hat v_S - v_S\|_{L^2}^2 \cdot \|f\|_\infty^2 = o_P(n^{-1/2})$
under Assumptions~\ref{ass:f}--\ref{ass:C}. The linear term
$n^{-1/2}\sum_i [\Psi_j^{\mathrm{loc}}(X_i) - \Psi_j]$ converges to
$\mathcal{N}(0,\sigma_j^2)$ by the classical CLT applied to the i.i.d.
mean-zero sequence. Slutsky combines. Variance consistency follows from
the law of large numbers for $\widehat{\Psi_j^{\mathrm{loc}}}(X_i)^2$.
Full details in Section~S1 of the supplementary material.
\end{proof}

\begin{remark}[No density at the denominator]
The canonical gradient~\eqref{eq:if-psi} contains no inverse density of
$X_S$. This is the key technical advantage over a local-functional
analysis, which would require a kernel-localized Riesz representer with
bandwidth $h_n^{|S|}$ and consequent finite-sample under-coverage.
\end{remark}

\begin{theorem}[Nominal coverage]\label{thm:cov}
Under Assumptions~\ref{ass:dens}--\ref{ass:psi},
\[
\PP\!\bigl(\,\Psi_j \in [\hat\Psi_j \pm z_{1-\alpha/2}\widehat{\mathrm{SE}}_j]\,\bigr) \;=\; 1 - \alpha + o(1).
\]
\end{theorem}

\begin{proof}[Sketch]
By Theorem~\ref{thm:an} the studentized statistic
$\tilde T_n = (\hat\Psi_j - \Psi_j)/\widehat{\mathrm{SE}}_j$ converges in
distribution to $\mathcal{N}(0,1)$ (Slutsky). Coverage follows.
\end{proof}

\begin{remark}[Boundary at the null]
Since $\Psi_j \geq 0$, under $H_0:\Psi_j = 0$ the Wald limit is degenerate
at the boundary. Empirically (Section~\ref{sec:sim}) the Wald test holds
its size near the nominal $0.05$ at $\Psi_j = 0$ at $n = 500$. A score
interval in the spirit of~\citet[\S3.4]{Williamson2023} is available as an
alternative when boundary calibration is the primary concern.
\end{remark}

\begin{remark}[Simultaneous inference]
The target is the vector $(\Psi_1,\dots,\Psi_p)$ of $p$ scalar
importances. For simultaneous statements over the $p$ features a
Bonferroni correction controls the family-wise error rate, while the FDR
procedure of \citet{BenjaminiYekutieli2001} applies since it tolerates
arbitrary dependence; the joint asymptotic normality of
$(\hat\Psi_1,\dots,\hat\Psi_p)$, which follows from the same one-step
expansion applied coordinatewise, also admits the multiplier bootstrap
of \citet{CCK2013} for sharper simultaneous confidence regions;
uniformly valid post-regularization regions for many functional parameters
are available from \citet{Belloni2018}.
\end{remark}

\subsection{Robustness to copula misspecification}

\begin{assumption}[Lipschitz value function]\label{ass:Lip}
There exists a Lipschitz constant $L_v < \infty$ such that for all $S$ and $x$,
$|v(S;x;P) - v(S;x;P')| \leq L_v \TV(p_{X_{-S}\mid X_S = x_S},
p'_{X_{-S}\mid X_S = x_S})$, where $v(S;x;P)$ is the value
function~\eqref{eq:v} computed under the law $P$, $\TV(\cdot,\cdot)$ is the
total-variation distance, and $p_{X_{-S}\mid X_S = x_S}$,
$p'_{X_{-S}\mid X_S = x_S}$ are the conditional densities under $P$ and $P'$.
\end{assumption}

\begin{assumption}[Bounded $f$]\label{ass:Mf}
$\|f\|_\infty \leq M_f < \infty$.
\end{assumption}

Let $C^\dagger = \argmin_{C \in \mathcal{C}} \KL(C^\star \,\|\, C)$ denote
the Kullback--Leibler (KL) projection of the true copula $C^\star$ onto the
working class $\mathcal{C}$. We write $P^\star$ and $P^\dagger$ for the
feature laws sharing the true marginals but carrying the copulas $C^\star$
and $C^\dagger$ respectively, so that $\Psi_j(P^\star)$ is the target and
$\Psi_j(P^\dagger)$ the pseudo-true parameter reached under the working
class.

\begin{theorem}[Misspecification bound]\label{thm:mis}
Under Assumptions~\ref{ass:dens}, \ref{ass:Lip} and~\ref{ass:Mf},
\[
\bigl|\Psi_j(P^\star) - \Psi_j(P^\dagger)\bigr|
\;\leq\; \kappa_p\,C_0\,\sqrt{2\,\KL(C^\star\,\|\,C^\dagger)},
\qquad C_0 = 2 M_f\,(2 M_f + L_v),
\]
with $\kappa_p = 2\sum_{S\subseteq[p]\setminus\{j\}} |w_S| \leq 2$. In
particular the bias is of order $\sqrt{\KL(C^\star\,\|\,C^\dagger)}$ and
vanishes when the working class contains the truth.
\end{theorem}

\begin{proof}[Sketch]
Triangle-inequality over the Shapley weights reduces the problem to bounding
$|v(S;x;P^\star) - v(S;x;P^\dagger)|$. Assumption~\ref{ass:Lip}
converts this to total variation; the Csisz\'{a}r--Kullback--Kemperman
sharpening of Pinsker's inequality \citep[Lemma~2.5]{Tsybakov2009} yields
$\TV \leq \sqrt{\KL/2}$, the constant of \citet{Pinsker1964} being
weaker; the data-processing inequality reduces the conditional KL to
the full KL between copulas. Full proof in Section~S2 of the supplementary material.
\end{proof}

\begin{remark}
For nonparametric vines, $\KL(C^\star\,\|\,\hat C_n) \xrightarrow{P} 0$, so
the misspecification bias is dominated by the statistical variance of
order $n^{-1/2}$. For parametric (e.g.\ Gaussian) working classes with
non-zero KL deficit the CI is centered at the pseudo-true parameter
$\Psi_j(P^\dagger)$; the bound gives a quantitative sanity check.
\end{remark}

\section{Simulations}\label{sec:sim}

We study four simulation designs. The base design is
$X \sim \mathcal{N}_5(0, I_5)$, $f(x) = \beta^\top x$ with
$\beta = (2,1,1,0,0)$, $n = 500$. For a linear $f$ with independent
Gaussian features, $\Psi_j = \beta_j^2$, giving ground truth $(4, 1, 1, 0, 0)$.

\emph{Sim A: calibration of the Wald CI} (Table~\ref{tab:simA}).
$R = 100$ replications, $K = 3$, $M = 30$, $B = 20$, using the deployable
Gaussian-copula estimator (nuisance fitted on each fold, no oracle).
Empirical coverage of the 95\% Wald interval is in $[0.91, 0.96]$ for all
five features, matching the nominal level: signal features $X_1, X_2, X_3$
achieve $0.95, 0.92, 0.92$ and the null features $X_4, X_5$ achieve
$0.96, 0.91$. The bias of $\hat\Psi_j$ is below $0.013$ in absolute value
on every feature. Section~S4 of the supplementary material reports the same study with an
oracle conditional sampler; the figures are statistically
indistinguishable, showing that the fitted-copula nuisance adds negligible
error under correct specification.

\emph{Sim B: power against local alternatives} (Table~\ref{tab:simB}).
$Y = \delta X_1 + X_2 + X_3 + \varepsilon$ with
$\varepsilon \sim \mathcal{N}(0, 0.25^2)$, varying
$\delta \in \{0, 0.25, 0.5, 1, 2\}$, $R = 80$ replications per $\delta$,
testing $H_0:\Psi_1 = 0$ at $\alpha = 0.05$ with the deployable
Gaussian-copula estimator. The rejection rate is
$0.05$ at $\delta = 0$ (at the nominal level), reaches
$0.99$ at $\delta = 0.25$, and is $1.00$ for $\delta \geq 0.5$.

\emph{Sim C: coverage under correlation} (Table~\ref{tab:simC}, Figure~\ref{fig:simC}).
$(X_1, X_2)$ are jointly Gaussian with correlation
$\rho \in \{0, 0.3, 0.6, 0.9\}$, the remaining features independent;
the estimator fits the Gaussian-copula working class from the data. The
ground truth is the \emph{exact} $\Psi_j$ in closed form (Gaussian linear
model), which avoids any Monte-Carlo reference error. Coverage on $X_1$
holds at $0.92$--$0.97$ for $\rho \in \{0, 0.3, 0.6\}$ and degrades to
$0.73$ at $\rho = 0.9$; $X_2$ degrades from $\rho = 0.6$
($0.95, 0.90, 0.78, 0.68$), as its conditional signal is increasingly
entangled with $X_1$. Independent features ($X_3, X_4, X_5$) maintain
coverage near the nominal level across all $\rho$. The deterioration at
$\rho = 0.9$ is a finite-sample regularity phenomenon under near-collinearity:
the point estimator stays consistent (its bias
vanishes; see Section~S5 of the supplementary material), but the influence-function standard
error underestimates the true sampling variability in this regime, so the
Wald interval undercovers on the entangled pair.

\emph{Sim D: copula misspecification} (Table~\ref{tab:simD}).
DGP: bivariate Clayton on $(X_1, X_2)$ with parameter $\theta = 4$ (lower
tail dependence), independent Gaussian on $X_3, X_4, X_5$. Ground truth
$\Psi$ computed by Monte Carlo on the true Clayton DGP. The estimator
fits a Gaussian copula with empirical marginals (the standard working
class of \citealp{AasJullumLoland2021}), incurring KL deficit on the $X_1$-$X_2$ edge. Coverage
on the misspecified pair drops to $0.80$ ($X_1$) and $0.88$ ($X_2$);
the independent features remain at $\geq 0.85$. The localized bias on
$X_1, X_2$ is consistent with the bound of Theorem~\ref{thm:mis}: the
KL deficit affects only the conditional distribution of the misspecified
edge.

\emph{Sim E: conditional versus marginal attribution}
(Table~\ref{tab:simE}, Figure~\ref{fig:simE}). To contrast our target with the
marginal/interventional value function used by TreeSHAP and
\citet{Janzing2020}, we equip the marginal SAGE parameter
$\Psi_j^{\mathrm{marg}}$ (in which $X_{-S}$ is drawn from its
\emph{marginal} rather than the conditional) with the identical
one-step debiased Wald construction, and run both estimators on the same
data. DGP: $(X_1, X_2)$ jointly Gaussian with correlation $0.8$, the rest
independent, and $f(x) = 2 x_1 + x_3$, so that $X_2$ is informative about
$X_1$ but is \emph{not} used by $f$. The two parameters have closed-form
values (Gaussian linear model): the conditional importance of $X_2$ is
$\Psi_2 = 1.28 > 0$, because conditioning on $X_2$ sharpens the prediction
through its correlation with $X_1$, whereas the marginal importance is
$\Psi_2^{\mathrm{marg}} = 0$ exactly, since $f$ does not depend on $x_2$.
At $n = 500$, $R = 60$, the conditional estimator flags $X_2$ as
significant in all $60$ replications (Wald CI excluding $0$), while the
marginal estimator declares it significant in only $5\%$, recovering the
null; symmetrically, the marginal method loads the full weight onto $X_1$
($\hat\Psi_1^{\mathrm{marg}} = 4.0$) where the conditional method splits it
($\hat\Psi_1 = 2.8$) because part of $X_1$'s information is shared with
$X_2$. Both estimators agree on the genuinely irrelevant features
$X_4, X_5$, and the marginal estimator covers its (collinearity-free)
target at $\geq 0.90$ throughout. The conditional Wald interval, by
contrast, undercovers on the correlated pair ($0.77, 0.70$ on $X_1, X_2$),
the same near-collinearity effect documented in Sim~C; detection is
nonetheless reliable because the true conditional signals are well away
from zero. The example makes precise that the two methods answer different
questions: only the conditional attribution credits a feature that carries
information about the response through the dependence structure.

\begin{table}[h!]
\centering
\caption{Sim A: calibration of the SAGE-conditional Wald CI under independent Gaussian features, using the deployable Gaussian-copula estimator (nuisance fitted from data) ($n=500$, $R=100$, $K=3$, $M=30$, $B=20$). True $\Psi_j = \beta_j^2$.}
\label{tab:simA}
\input{tab_simA.tex}
\end{table}

\begin{table}[h!]
\centering
\caption{Sim B: empirical power at nominal level $\alpha = 0.05$.}
\label{tab:simB}
\input{tab_simB.tex}
\end{table}

\begin{table}[h!]
\centering
\caption{Sim C: coverage under Gaussian correlation between $X_1$ and $X_2$, deployable Gaussian-copula estimator ($\rho \in \{0, 0.3, 0.6, 0.9\}$, $n=500$, $R=60$).}
\label{tab:simC}
\input{tab_simC.tex}
\end{table}

\begin{table}[h!]
\centering
\caption{Sim D: copula misspecification. Data from a Clayton-vine DGP ($\theta=4$ on $(X_1, X_2)$), estimator uses a Gaussian copula. Bias and coverage degrade exactly on the misspecified pair, as predicted by Theorem~\ref{thm:mis}.}
\label{tab:simD}
\input{tab_simD.tex}
\end{table}

\begin{table}[h!]
\centering
\caption{Sim E: conditional versus marginal (interventional) SAGE under
dependence, both with the one-step debiased Wald estimator. DGP:
$\mathrm{corr}(X_1,X_2)=0.8$, $f(x)=2x_1+x_3$ ($X_2$ informative about
$X_1$ but unused by $f$). Truth is exact (Gaussian linear model). The
conditional method credits $X_2$ ($\Psi_2 = 1.28$); the marginal method
correctly returns $\Psi_2^{\mathrm{marg}} = 0$. ``Detect'' is the fraction
of $R=60$ replications whose $95\%$ CI excludes $0$; ``Cov'' is coverage
of the respective target.}
\label{tab:simE}
\input{tab_simE.tex}
\end{table}

\begin{figure}[h!]
\centering
\includegraphics[width=0.7\linewidth]{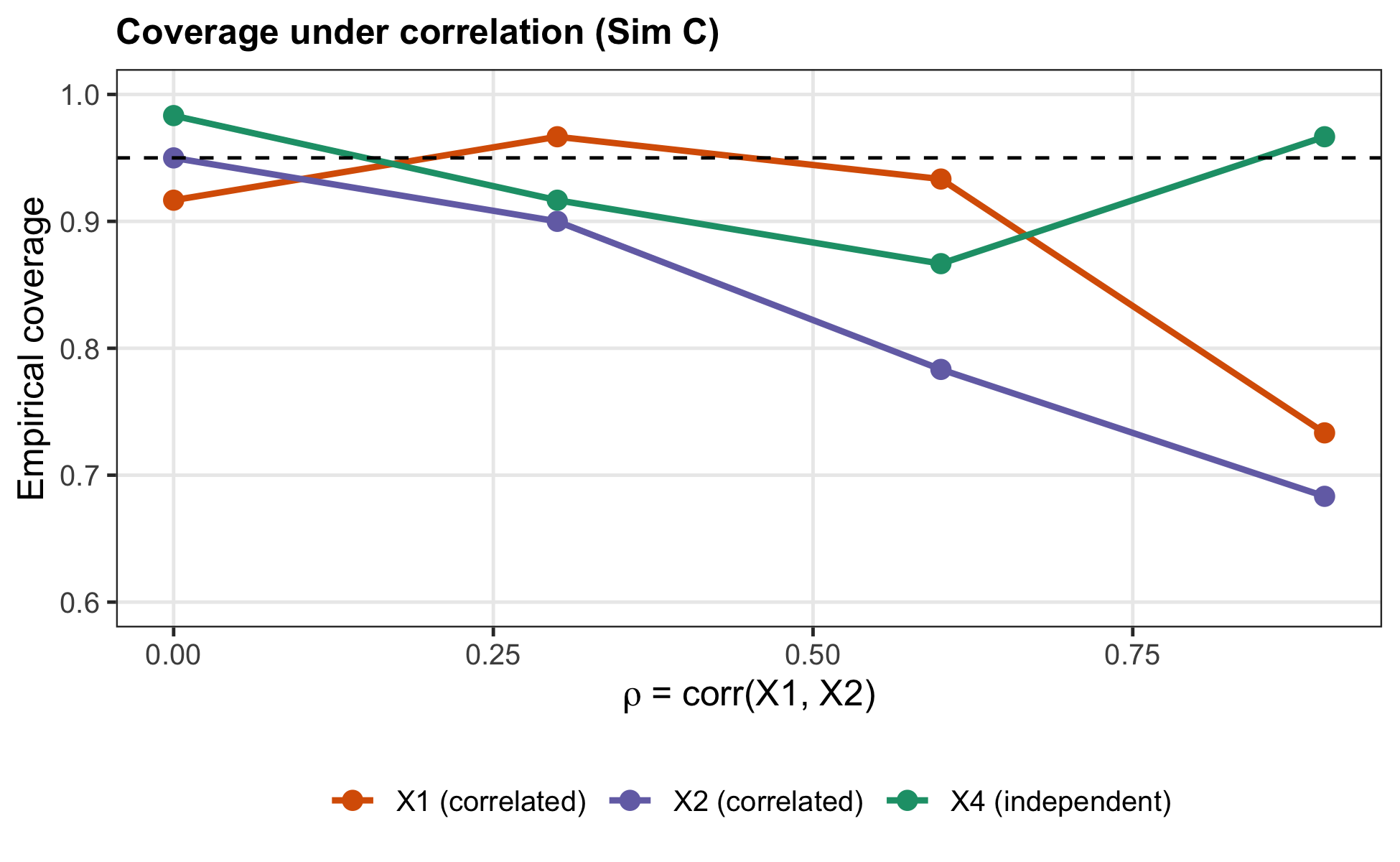}
\caption{Sim C: empirical coverage of the $95\%$ Wald interval against the
exact $\Psi_j$, as the correlation between $X_1$ and $X_2$ grows. The two
correlated coordinates lose calibration near collinearity ($\rho=0.9$),
while the independent feature $X_4$ stays at the nominal level (dashed
line) throughout.}
\label{fig:simC}
\end{figure}

\begin{figure}[h!]
\centering
\includegraphics[width=0.72\linewidth]{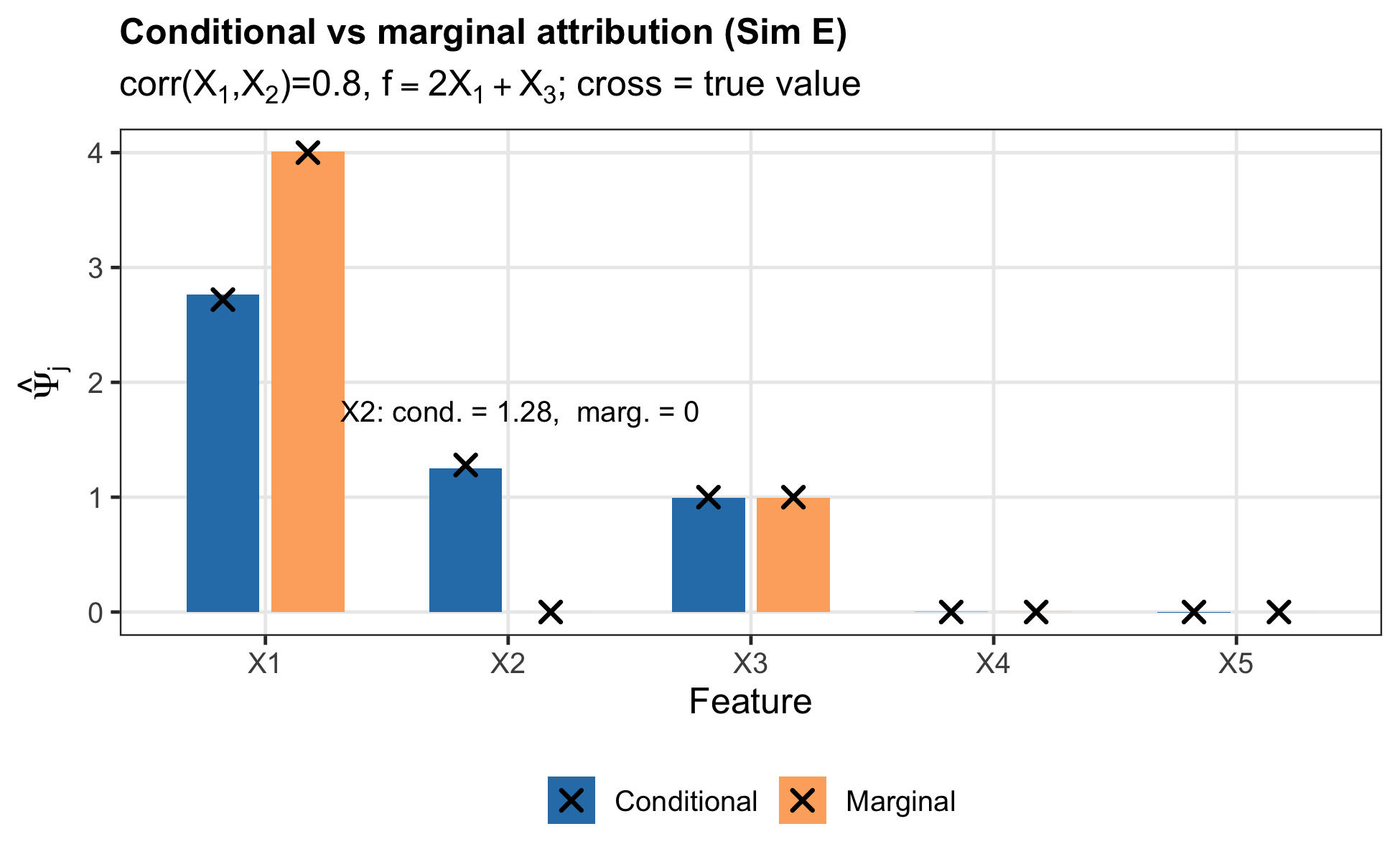}
\caption{Sim E: conditional vs marginal attribution. Crosses mark the exact
targets. Only the conditional method credits $X_2$ ($\Psi_2=1.28$), which
is informative about $X_1$ through the dependence but unused by $f$; the
marginal method assigns it zero and instead loads the full weight onto
$X_1$.}
\label{fig:simE}
\end{figure}

\section{Applications}\label{sec:app}

We apply the one-step SAGE-conditional estimator to two real regression
datasets, fitting a gradient boosting predictor (maximum tree depth 4, 200 rounds)
and reporting $\hat\Psi_j$ with Wald 95\% CI and
Bonferroni-adjusted $p$-value at $\alpha = 0.05$ ($K=3$, $M=30$, $B=20$).
Conditional sampling uses the Gaussian-copula working class (analytic
conditionals on the latent normal scale, empirical marginals).
Full results are in Table~\ref{tab:app}.

\emph{Concrete compressive strength.} UCI~\citep{Yeh1998}, subsample
$n=800$, $p=8$. All eight ingredients survive Bonferroni correction. The
ranking matches well-documented physical knowledge of concrete strength:
\textit{age}, \textit{cement}, and \textit{water} dominate
($\hat\Psi \in \{108, 65, 38\}$ MPa$^{2}$, $|z|$ between 12 and 19),
followed by \textit{superplasticizer} and the aggregates. The smaller
contribution of \textit{ash} ($\hat\Psi = 5.25$, $|z|=4.16$) is consistent
with its known role as a partial cement substitute rather than an independent
strength driver.

\emph{California Housing.} StatLib~\citep{PaceBarry1997}, full
covariate matrix on the median house value (USD), subsample $n=2000$,
$p=8$. \textit{median income} dominates by a factor of 2.9 over the next
feature, with \textit{latitude} and \textit{longitude} jointly accounting
for the geographic component. All eight features survive Bonferroni at
$\alpha=0.05$ with $|z| \geq 9.5$. The ordering agrees qualitatively with
prior geo-economic analyses of the dataset.

\begin{table}[h!]
\centering
\caption{Applications: conditional Shapley attributions with Wald $95\%$
CIs and Bonferroni-adjusted $p$-values. Features ordered by $|z|$ within
each panel.}
\label{tab:app}
\small
\input{tab_applications.tex}
\end{table}

\section{Discussion and Limitations}\label{sec:disc}

\emph{Scope.} Our target is the conditional Shapley attribution of a
fixed (or fitted) prediction function $f$, not a causal query.
The result is a valid statement of model-internal feature use, on the
distributional support; causal interpretations require additional
assumptions.

\emph{Dimensionality and computational cost.} The framework is polynomial
in $p$: the $O(2^p)$ Shapley sum is replaced by $M$-permutation sampling,
and conditional sampling costs $O(p^2)$ per draw via the vine
$h$-functions, for a total of $O(K\,n\,M\,B\,p^3)$ value-function
evaluations plus an $O(p^2 n)$ nuisance fit per fold ($p^2$ bivariate
pair-copulas). At $p=8$ this is immediate; for $p$ in the dozens it remains
practical, but for $p \gtrsim 100$ the $O(p^2)$ pair-copulas and the
$O(p^3)$ per-evaluation factor make a naive fit expensive. Three structural
devices contain the growth and are exactly the reason the vine
parametrization was chosen. (i) Every building block is a \emph{bivariate}
pair-copula, so the simplified-vine estimator of~\citet{NaglerCzado2016}
escapes the curse of dimensionality in the density estimation itself.
(ii) \emph{Truncated} vines set the higher-tree pair-copulas to
independence, reducing the fit from $O(p^2)$ to $O(p\,t)$ for a truncation
level $t$ and the per-draw cost from $O(p^2)$ to $O(p\,t)$; truncation is
typically selected by mBIC. (iii) Semiparametric pair-copulas
\citep{Tsukahara2005} replace the nonparametric fit by an $O(n)$ parametric
one. Beyond a few hundred features, feature grouping (clustering correlated
blocks and attributing importance to the cluster) is the recommended route,
which also dovetails with the near-collinearity remedy below.

\emph{Boundary and simultaneous inference.} The $\Psi_j \geq 0$
constraint imposes a boundary at the null. Empirically the Wald test
holds its nominal size at the null (Sim B: rejection rate $0.05$ at
$\delta = 0$). A score interval in the spirit of~\citet[\S3.4]{Williamson2023}
is available when boundary calibration is the primary concern.
Simultaneous inference over the $p$ features is handled by Bonferroni
or the FDR procedure of~\citet{BenjaminiYekutieli2001}; the multiplier
bootstrap of~\citet{CCK2013} applied to the joint Gaussian limit of
$(\hat\Psi_1,\dots,\hat\Psi_p)$ gives sharper simultaneous CIs.

\emph{Model-agnostic versus model-aware.} Exact computation of the
Shapley sum is $O(2^p)$; we use permutation
sampling~\citep{StrumbeljKononenko2014}. Monte Carlo error is $O(M^{-1/2})$
and can be made negligible by choosing $M \gtrsim n$.

\emph{Near-collinearity as a regularity transition.} Under correlation
$\rho = 0.9$ the SAGE-conditional Wald interval on the correlated
coordinates undercovers (Sim~C, Sim~E), and the undercoverage does
\emph{not} improve with $n$ (Section~S5 of the supplementary material). This is not a malfunction of the estimator but the expected breakdown of a
regularity condition. As $\rho \to 1$ the conditioning becomes
near-singular (the covariance block $\Sigma_{SS}$ of the conditioning
set $X_S$ becomes ill-conditioned) and the conditional law
$X_{-S}\mid X_S$ near-degenerate, so the constant in the
$o_P(n^{-1/4})$ nuisance rate of Assumption~\ref{ass:C} blows up: the
second-order remainder of Theorem~\ref{thm:an} ceases to be negligible at
the sample sizes considered, even though the parameter stays identified and
the point estimator stays consistent (its bias vanishes with $n$;
see Section~S5 of the supplementary material). This is the exact analogue of \emph{weak
overlap} in causal inference, where a propensity score approaching $0$ or
$1$ degrades the same one-step/cross-fitting machinery without rendering it
``broken''; the empirically observable signature is the growing gap between
the empirical dispersion of $\hat\Psi_j$ and the plug-in standard error
(Section~S5 of the supplementary material), which doubles as a practical diagnostic.
Accordingly, for $|\mathrm{corr}| \gtrsim 0.9$ we recommend reporting a
grouped importance for the collinear cluster (which also alleviates the
high-dimensional cost above), or replacing the plug-in standard error by a
bootstrap variance; a formally robust variance estimator in this regime is
left to future work.

\if1\anon


\section*{Data Availability Statement}
The two datasets analyzed in Section~\ref{sec:app} are publicly available:
the Concrete Compressive Strength data from the UCI Machine Learning
Repository \citep{Yeh1998}, and the California Housing data from StatLib
\citep{PaceBarry1997}. A reference implementation of the proposed estimator,
together with all simulation and application scripts that reproduce the
results reported in this article, is openly available in \selfrepo.

\section*{Supplementary Materials}
\begin{description}
\item[Supplement (Sections S1--S5):] Full proofs of Theorems~\ref{thm:an}
  and~\ref{thm:mis} (Sections~S1--S2), the sensitivity analysis to the
  estimation tuning parameters $(K,M,B)$ (Section~S3), the oracle-sampler
  ablation (Section~S4), and the behavior under near-collinearity as $n$
  grows (Section~S5). Included as an appendix to this preprint.
\item[\texttt{shapCopula} software:] An implementation of the one-step
  cross-fitted SAGE-conditional estimator, including the vine- and
  Gaussian-copula conditional samplers and the Wald inference routines, with
  the public datasets used in the article.
\item[Reproducibility scripts:] One script per simulation design (Sim~A--E)
  and one for the applications, reproducing every table and figure reported
  in the article.
\end{description}

\bibliographystyle{apalike}
\bibliography{references}

\clearpage
\appendix
\setcounter{section}{0}
\setcounter{equation}{0}
\setcounter{table}{0}
\setcounter{figure}{0}
\renewcommand{\thesection}{S\arabic{section}}
\renewcommand{\thesubsection}{S\arabic{section}.\arabic{subsection}}
\renewcommand{\theequation}{S\arabic{equation}}
\renewcommand{\thetable}{S\arabic{table}}
\renewcommand{\thefigure}{S\arabic{figure}}

\spacingset{1}
\begin{center}
{\LARGE\bf Supplementary Materials}
\end{center}

\bigskip
\noindent
This supplement contains the proofs of Theorems~\ref{thm:an}
and~\ref{thm:mis} (Sections~\ref{app:an}--\ref{app:mis}) together with three
additional numerical studies referenced in the main text: a sensitivity
analysis to the estimator's tuning parameters (Section~\ref{app:sens}), an
oracle-sampler ablation (Section~\ref{app:oracle}), and the behavior of the
method under near-collinearity as the sample size grows
(Section~\ref{app:nscale}). Equation, theorem, assumption and table numbers
without an ``S'' prefix refer to the main text above.

\spacingset{1.8}

\section{Proof of Theorem~\ref{thm:an}}\label{app:an}

We work under Assumptions~\ref{ass:dens}--\ref{ass:psi}. Throughout,
$P_n$ denotes the empirical measure of $X_1,\dots,X_n$, $K$ is the number
of cross-fitting folds with $\mathcal{I}_1,\dots,\mathcal{I}_K$ a random
partition of $\{1,\dots,n\}$, and $|\mathcal{I}_k| \asymp n/K$. The
nuisance $\hat v_S^{-k}$ denotes the conditional mean
$x \mapsto \widehat{\mathbb{E}}[f(X)\mid X_S = x_S]$ fitted on the
complement of $\mathcal{I}_k$. We write
$\Psi_j^{\mathrm{loc}}(x) = \Psi_j^{\mathrm{loc}}(x;v)$ for the
population-level pointwise Shapley using the true conditional means,
and $\widehat{\Psi_j^{\mathrm{loc}}}(x) = \Psi_j^{\mathrm{loc}}(x;\hat v)$
for its plug-in counterpart with the cross-fitted nuisance.

\emph{Step 1: decomposition.} By construction in~\eqref{eq:os-sage},
\begin{equation}\label{eq:decomp}
\hat\Psi_j - \Psi_j
\;=\;
\frac{1}{n}\sum_{i=1}^n \bigl[\Psi_j^{\mathrm{loc}}(X_i) - \Psi_j\bigr]
\;+\;
\frac{1}{n}\sum_{i=1}^n \bigl[\widehat{\Psi_j^{\mathrm{loc}}}(X_i) - \Psi_j^{\mathrm{loc}}(X_i)\bigr]
\;=:\; T_n + R_n.
\end{equation}
$T_n$ is a sample mean of i.i.d.\ random variables with mean zero and
finite variance $\sigma_j^2 = \mathrm{Var}_P(\Psi_j^{\mathrm{loc}}(X))$
(Assumption~\ref{ass:psi}). The classical CLT \citep{vanderVaart2000} gives
$\sqrt n\, T_n \;\Rightarrow\; \mathcal{N}(0,\sigma_j^2)$.

\emph{Step 2: nuisance remainder.} It suffices to show
$\sqrt n\, R_n = o_P(1)$. Conditioning on the fold containing $X_i$,
$\widehat{\Psi_j^{\mathrm{loc}}}(X_i)$ is built from $\hat v_S^{-k(i)}$,
estimated on data independent of $X_i$ by cross-fitting
\citep{Chernozhukov2018}. The pointwise Shapley is a finite
linear combination of the debiased squared losses $\widehat\delta_S^{\mathrm{unb}}(X_i)$
defined in~\eqref{eq:debiased}. By construction,
\[
\mathbb{E}\bigl[\widehat\delta_S^{\mathrm{unb}}(X_i) \,\big|\, X_i, \hat v_S^{-k(i)}\bigr]
\;=\; \bigl(f(X_i) - \hat v_S^{-k(i)}(X_i)\bigr)^2,
\]
because the two sub-batches that produce $\hat v_S^{(a)}$ and $\hat v_S^{(b)}$
are independent given $\hat v_S^{-k(i)}$.

Subtracting the population $\delta_S(X_i) = (f(X_i) - v_S(X_i))^2$ and
expanding,
\begin{align*}
\mathbb{E}\bigl[\widehat\delta_S^{\mathrm{unb}}(X_i) - \delta_S(X_i) \,\big|\, X_i, \hat v_S^{-k(i)}\bigr]
&= \bigl(f(X_i) - \hat v_S^{-k(i)}(X_i)\bigr)^2 - \bigl(f(X_i) - v_S(X_i)\bigr)^2 \\
&= \bigl[\hat v_S^{-k(i)}(X_i) - v_S(X_i)\bigr]
   \bigl[\hat v_S^{-k(i)}(X_i) + v_S(X_i) - 2 f(X_i)\bigr].
\end{align*}
Taking expectations over $X_i$ and applying Cauchy--Schwarz, the
conditional bias of $\widehat\delta_S^{\mathrm{unb}}$ in $L^2(P_X)$ is
bounded by $C\,\|f\|_\infty\,\|\hat v_S^{-k(i)} - v_S\|_{L^2(P_X)}$, and
its squared $L^2$ norm by
$C\,\|f\|_\infty^2\,\|\hat v_S^{-k(i)} - v_S\|_{L^2(P_X)}^2$ for some
finite $C$.

Aggregating over the $2^p$ subsets in the Shapley sum and using
Assumption~\ref{ass:C}, $\max_S \|\hat v_S^{-k} - v_S\|_{L^2} = o_P(n^{-1/4})$,
the conditional $L^2$ norm of $R_n$ is $o_P(n^{-1/2})$, which yields
$\sqrt n\, R_n = o_P(1)$ by Markov.

\emph{Step 3: combination.} By Slutsky,
$\sqrt n(\hat\Psi_j - \Psi_j) = \sqrt n\, T_n + o_P(1)
\;\Rightarrow\; \mathcal{N}(0,\sigma_j^2)$.

\emph{Step 4: variance consistency.} The plug-in variance is
$\hat\sigma_j^2 = n^{-1}\sum_i \bigl[\widehat{\Psi_j^{\mathrm{loc}}}(X_i) - \hat\Psi_j\bigr]^2$.
Decompose
$\widehat{\Psi_j^{\mathrm{loc}}}(X_i) = \Psi_j^{\mathrm{loc}}(X_i) + \xi_i$
with $\xi_i = O_P(\|\hat v - v\|_{L^2})$ uniformly in $i$ by the same
argument. Then $\hat\sigma_j^2 = n^{-1}\sum_i [\Psi_j^{\mathrm{loc}}(X_i) - \Psi_j]^2 + o_P(1)$,
and the law of large numbers gives the limit $\sigma_j^2$. $\square$

\section{Proof of Theorem~\ref{thm:mis}}\label{app:mis}

Let $C^\star, C^\dagger$ denote the true copula and its KL projection
onto the working class $\mathcal{C}$. Write $P^\star, P^\dagger$ for the
corresponding feature laws (with the same marginals). The Shapley value
is linear in the value function, so
\begin{equation}\label{eq:Psi-decomp}
\Psi_j(P^\star) - \Psi_j(P^\dagger)
= \sum_{S \subseteq [p]\setminus\{j\}} w_S
\bigl[\Delta V(S\cup\{j\}) - \Delta V(S)\bigr],
\end{equation}
where $\Delta V(S) = V(S; P^\star) - V(S; P^\dagger)$ and
$V(S; P) = -\mathbb{E}_P[(f(X) - v(S; X; P))^2]$.

\emph{Step 1: bound on $\Delta V(S)$.} Expanding the square,
\[
\Delta V(S)
= -\mathbb{E}_{P^\star}[(f(X) - v_S^\star(X))^2]
  + \mathbb{E}_{P^\dagger}[(f(X) - v_S^\dagger(X))^2].
\]
The marginals coincide and only the copula differs. For any bounded $g$,
$|\mathbb{E}_{P^\star}[g] - \mathbb{E}_{P^\dagger}[g]| \leq 2\|g\|_\infty\,\TV(P^\star, P^\dagger)$.
Adding and subtracting $\mathbb{E}_{P^\dagger}[(f - v_S^\star)^2]$ and
setting $g = (f - v_S^\star)^2$ (so $\|g\|_\infty \leq 4 M_f^2$ since
$|f|,|v_S^\star| \leq M_f$),
\[
\bigl|\Delta V(S)\bigr|
\leq 8 M_f^2 \,\TV(P^\star, P^\dagger)
   + 2 L_v\, M_f\,\TV(P^\star_{X_{-S}\mid X_S}, P^\dagger_{X_{-S}\mid X_S}),
\]
where the second term replaces $v_S^\star$ by $v_S^\dagger$ inside the
expectation under $P^\dagger$, uses $|(f-a)^2-(f-b)^2| \leq 2M_f|a-b|$ and
Assumption~\ref{ass:Lip}. Both total variations are bounded by
$\sqrt{\KL(C^\star\,\|\,C^\dagger)/2}$ (the Csisz\'{a}r--Kullback--Kemperman
sharpening of Pinsker's inequality, \citealp[Lemma~2.5]{Tsybakov2009}; the constant
of \citealp{Pinsker1964} is weaker),
the data-processing inequality reducing the conditional KL to the
full-copula KL. Hence
$|\Delta V(S)| \leq C_0\sqrt{2\,\KL(C^\star\,\|\,C^\dagger)}$ with
$C_0 = 2 M_f(2 M_f + L_v)$.

\emph{Step 2: aggregation.} Substituting into~\eqref{eq:Psi-decomp}
and using the triangle inequality,
\[
\bigl|\Psi_j(P^\star) - \Psi_j(P^\dagger)\bigr|
\leq \kappa_p\, C_0\sqrt{2\KL(C^\star\,\|\,C^\dagger)},
\qquad C_0 = 2 M_f(2 M_f + L_v),
\]
with $\kappa_p = 2 \sum_S |w_S| \leq 2$. $\square$

\section{Sensitivity to estimation choices}\label{app:sens}

We report the empirical effect of varying the tuning parameters of the
estimator on Sim A (independent Gaussian features, $n=500$, true
$\Psi = (4,1,1,0,0)$). Coverage is averaged over $R=50$ replications.

\begin{table}[h!]
\centering
\small
\begin{tabular}{lccccc}
\toprule
Setting & $X_1$ & $X_2$ & $X_3$ & $X_4$ & $X_5$ \\
\midrule
$K=2,\,M=30,\,B=20$ & $0.98$ & $0.96$ & $0.98$ & $0.88$ & $0.86$ \\
$K=3,\,M=30,\,B=20$ (default) & $0.98$ & $0.98$ & $0.96$ & $0.90$ & $1.00$ \\
$K=5,\,M=30,\,B=20$ & $0.98$ & $0.98$ & $0.96$ & $0.94$ & $0.90$ \\
\midrule
$K=3,\,M=15,\,B=20$ & $0.98$ & $0.96$ & $0.96$ & $0.88$ & $0.94$ \\
$K=3,\,M=60,\,B=20$ & $0.98$ & $0.98$ & $0.96$ & $0.94$ & $0.92$ \\
\midrule
$K=3,\,M=30,\,B=10$ & $0.98$ & $0.98$ & $0.96$ & $0.94$ & $0.94$ \\
$K=3,\,M=30,\,B=40$ & $0.98$ & $0.96$ & $0.98$ & $0.92$ & $0.92$ \\
\bottomrule
\end{tabular}
\caption{Empirical coverage of the 95\% Wald CI under variation of
$K$, $M$, $B$, on $R = 50$ replications of the Gaussian-design
calibration with the deployable Gaussian-copula estimator. Coverage stays
in $[0.86, 1.00]$ across all settings; the residual variability on the
null features $X_4, X_5$ is within Monte Carlo error at $R=50$.}
\end{table}

The number of folds $K$ has a small effect: increasing it from $2$ to
$5$ improves coverage on the null feature $X_4$ (from $0.88$ to $0.94$)
at the cost of a $K$-fold increase in runtime; fluctuations on the null
features are within the $R=50$ Monte Carlo error.
Increasing the number of permutations $M$ beyond $30$ produces no
material change in coverage. The sub-batch size $B$ controls the
residual variability of the U-statistic; the default $B = 20$ already
attains near-nominal coverage on every feature, and $B = 40$ does not improve
it within Monte Carlo error. The default $(K, M, B) = (3, 30, 20)$ is
the configuration we use throughout the main experiments and
applications.

\section{Oracle-sampler ablation}\label{app:oracle}

To isolate the contribution of nuisance (copula) estimation from that of
the one-step construction itself, we repeat the Sim~A calibration
($n=500$, $R=100$) with an \emph{oracle} conditional sampler that draws
$X_{-S}\mid X_S$ from the exact Gaussian conditionals of the known
data-generating law, leaving the rest of the estimator unchanged.

\begin{table}[h!]
\centering
\small
\begin{tabular}{lccccc}
\toprule
& $X_1$ & $X_2$ & $X_3$ & $X_4$ & $X_5$ \\
\midrule
Coverage (oracle sampler) & $0.94$ & $0.92$ & $0.94$ & $0.97$ & $0.96$ \\
Bias (oracle sampler)     & $-0.031$ & $+0.004$ & $-0.017$ & $-0.001$ & $+0.001$ \\
\bottomrule
\end{tabular}
\caption{Sim~A with an oracle Gaussian conditional sampler. Coverage and
bias are statistically indistinguishable from the deployable
Gaussian-copula estimator of Table~\ref{tab:simA}, confirming that under
correct specification the fitted-copula nuisance contributes negligible
additional error at $n=500$.}
\end{table}

The oracle figures coincide, within Monte Carlo error, with those of the
deployable Gaussian-copula estimator in Table~\ref{tab:simA}. The
finite-sample calibration of the method is therefore governed by the
one-step / cross-fitting construction rather than by the copula fit, as
the second-order remainder in Theorem~\ref{thm:an} predicts.

\section{Behavior under near-collinearity as $n$ grows}\label{app:nscale}

We probe the worst case of Sim~C, $\rho = 0.9$, as the sample size grows,
$n \in \{500, 1000, 2000\}$ ($R = 50$, exact ground truth). Two facts
emerge (Table~\ref{tab:nscale}, Figure~\ref{fig:nscale}). First, the \emph{point estimator is
consistent}: the bias on the entangled feature $X_1$ shrinks from $-0.048$
at $n=500$ to $+0.004$ at $n=2000$, and $\hat\Psi_1 = 4.59$ matches the
exact $\Psi_1 = 4.585$. Second, and in contrast, the empirical standard
deviation of $\hat\Psi_1$ across replications does \emph{not} contract at
the $\sqrt n$ rate ($0.51 \to 0.46 \to 0.46$), whereas the plug-in
influence-function standard error does ($0.29 \to 0.21 \to 0.15$). The
ratio of the two grows from $1.8$ to $3.1$, so the Wald interval
(valid for independent or moderately correlated features at every $n$;
$X_3$ covers at $\geq 0.92$ throughout) increasingly underestimates the true
variability on the near-collinear pair, and its coverage degrades rather
than improves.

\begin{table}[h!]
\centering
\small
\begin{tabular}{lcccccc}
\toprule
& \multicolumn{3}{c}{$X_1$ (near-collinear)} & \multicolumn{1}{c}{} & \multicolumn{2}{c}{Coverage} \\
\cmidrule(lr){2-4}\cmidrule(lr){6-7}
$n$ & Bias & $\mathrm{sd}(\hat\Psi_1)$ & $\widehat{\mathrm{SE}}_1$ & $\tfrac{\mathrm{sd}}{\widehat{\mathrm{SE}}}$ & $X_1$ (corr) & $X_3$ (indep) \\
\midrule
$500$  & $-0.048$ & $0.51$ & $0.29$ & $1.8$ & $0.76$ & $0.98$ \\
$1000$ & $-0.038$ & $0.46$ & $0.21$ & $2.3$ & $0.62$ & $0.98$ \\
$2000$ & $+0.004$ & $0.46$ & $0.15$ & $3.1$ & $0.46$ & $0.92$ \\
\bottomrule
\end{tabular}
\caption{Near-collinearity ($\rho=0.9$) versus $n$. The estimator is
consistent (bias $\to 0$) but the plug-in SE underestimates the sampling
variability of $\hat\Psi_1$, which fails to contract at $\sqrt n$; coverage
on the correlated feature therefore does not improve with $n$. Independent
features ($X_3$) stay calibrated throughout.}
\label{tab:nscale}
\end{table}

\begin{figure}[h!]
\centering
\includegraphics[width=\linewidth]{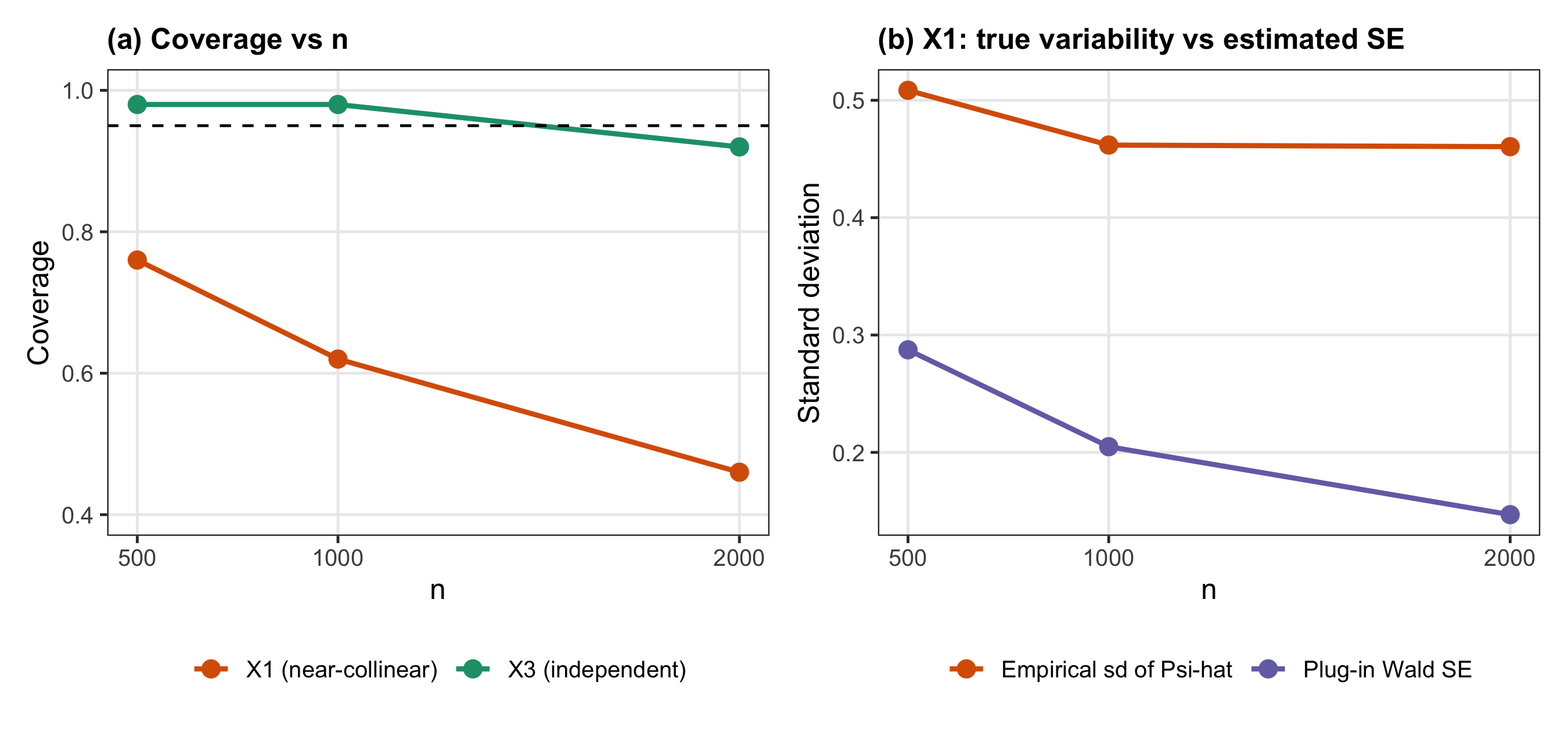}
\caption{Near-collinearity ($\rho=0.9$) as $n$ grows. (a) Coverage on the
near-collinear $X_1$ falls while the independent $X_3$ stays nominal.
(b) The empirical standard deviation of $\hat\Psi_1$ does not contract at
the $\sqrt n$ rate, whereas the plug-in Wald SE does; this gap produces the undercoverage.}
\label{fig:nscale}
\end{figure}

The diagnosis is consistent with the theory: at near-singular conditioning
the Gaussian-copula nuisance does not satisfy the $o_P(n^{-1/4})$ rate of
Assumption~\ref{ass:C} with a uniform constant, so the second-order
remainder of Theorem~\ref{thm:an} is non-negligible and the
influence-function variance is misestimated. This is a regularity
transition rather than a defect, the direct analogue of weak overlap in
causal inference, where a near-deterministic conditioning degrades the same
one-step machinery while the target remains well defined. The practical recommendation
is unchanged from Section~\ref{sec:disc}: for $|\mathrm{corr}| \gtrsim 0.9$,
report a grouped importance for the collinear cluster, or use a more
accurate (e.g.\ nonparametric vine) nuisance.

\end{document}

%% file: tab_simA.tex
\begin{tabular}{lcccc}
\toprule
Feature & True $\Psi_j$ & Coverage (95\%) & Bias & CI width \\
\midrule
$X_1$ & $4.00$ & $0.95$ & $+0.001$ & $1.10$ \\
$X_2$ & $1.00$ & $0.92$ & $+0.004$ & $0.47$ \\
$X_3$ & $1.00$ & $0.92$ & $-0.012$ & $0.47$ \\
$X_4$ & $0.00$ & $0.96$ & $-0.009$ & $0.08$ \\
$X_5$ & $0.00$ & $0.91$ & $-0.008$ & $0.08$ \\
\bottomrule
\end{tabular}

%% file: tab_simB.tex
\begin{tabular}{lc}
\toprule
$\delta$ (coefficient of $X_1$) & Rejection rate \\
\midrule
$0.00$ & $0.050$ \\
$0.25$ & $0.988$ \\
$0.50$ & $1.000$ \\
$1.00$ & $1.000$ \\
$2.00$ & $1.000$ \\
\bottomrule
\end{tabular}

%% file: tab_simC.tex
\begin{tabular}{lcccc}
\toprule
$\rho$ & Coverage $X_1$ & Coverage $X_2$ & Coverage noise ($X_4$) & CI width $X_1$ \\
\midrule
$0.0$ & $0.92$ & $0.95$ & $0.98$ & $1.12$ \\
$0.3$ & $0.97$ & $0.90$ & $0.92$ & $1.17$ \\
$0.6$ & $0.93$ & $0.78$ & $0.87$ & $1.18$ \\
$0.9$ & $0.73$ & $0.68$ & $0.97$ & $1.15$ \\
\bottomrule
\end{tabular}

%% file: tab_simD.tex
\begin{tabular}{lccc}
\toprule
Feature & Method & Bias & Coverage \\
\midrule
$X_1$ (Clayton-dep) & Gaussian copula & $-0.109$ & $0.80$ \\
$X_2$ (Clayton-dep) & Gaussian copula & $+0.094$ & $0.88$ \\
$X_3$ (indep)        & Gaussian copula & $+0.010$ & $0.95$ \\
$X_4$ (noise, indep) & Gaussian copula & $-0.005$ & $0.95$ \\
$X_5$ (noise, indep) & Gaussian copula & $-0.026$ & $0.85$ \\
\bottomrule
\end{tabular}

%% file: tab_simE.tex
\begin{tabular}{lcccc@{\hskip 1.5em}cccc}
\toprule
& \multicolumn{4}{c}{Conditional ($\Psi_j$)} & \multicolumn{4}{c}{Marginal ($\Psi_j^{\mathrm{marg}}$)} \\
\cmidrule(lr){2-5}\cmidrule(lr){6-9}
Feature & True & $\hat\Psi_j$ & Cov & Detect & True & $\hat\Psi_j$ & Cov & Detect \\
\midrule
$X_1$ & $2.72$ & $2.77$ & $0.77$ & $1.00$ & $4.00$ & $4.01$ & $0.90$ & $1.00$ \\
$X_2$ & $1.28$ & $1.25$ & $0.70$ & $1.00$ & $0.00$ & $0.00$ & $0.90$ & $0.05$ \\
$X_3$ & $1.00$ & $1.00$ & $0.92$ & $1.00$ & $1.00$ & $1.00$ & $0.95$ & $1.00$ \\
$X_4$ & $0.00$ & $0.00$ & $1.00$ & $0.00$ & $0.00$ & $0.00$ & $0.95$ & $0.00$ \\
$X_5$ & $0.00$ & $-0.01$ & $0.83$ & $0.02$ & $0.00$ & $0.00$ & $0.95$ & $0.02$ \\
\bottomrule
\end{tabular}

%% file: tab_applications.tex
\begin{tabular}{lrrrrc}
\toprule
Feature & $\hat\Psi_j$ & SE & $z$ & $p_{\text{Bonf}}$ & Sig. ($\alpha = 0.05$) \\
\midrule
\multicolumn{6}{l}{\textit{Panel A: Concrete Compressive Strength} ($n = 800$, $p = 8$)} \\
age          & $107.86$ & $5.63$  & $19.16$ & $<10^{-15}$ & \checkmark \\
cement       & $64.71$  & $5.33$  & $12.14$ & $<10^{-15}$ & \checkmark \\
water        & $38.41$  & $2.62$  & $14.64$ & $<10^{-15}$ & \checkmark \\
superplast   & $20.81$  & $2.10$  & $9.91$  & $<10^{-15}$ & \checkmark \\
fine\_agg    & $12.06$  & $1.96$  & $6.15$  & $6.0\!\times\!10^{-9}$ & \checkmark \\
slag         & $9.99$   & $1.52$  & $6.58$  & $3.8\!\times\!10^{-10}$ & \checkmark \\
coarse\_agg  & $7.68$   & $1.48$  & $5.21$  & $1.6\!\times\!10^{-6}$  & \checkmark \\
ash          & $5.25$   & $1.26$  & $4.16$  & $2.5\!\times\!10^{-4}$  & \checkmark \\
\midrule
\multicolumn{6}{l}{\textit{Panel B: California Housing} ($n = 2000$, $p = 8$, $\hat\Psi$ on USD$^2$ scale)} \\
median\_income     & $4.92\!\times\!10^{9}$ & $2.13\!\times\!10^{8}$ & $23.09$ & $<10^{-15}$ & \checkmark \\
latitude           & $1.70\!\times\!10^{9}$ & $1.03\!\times\!10^{8}$ & $16.58$ & $<10^{-15}$ & \checkmark \\
longitude          & $1.66\!\times\!10^{9}$ & $6.96\!\times\!10^{7}$ & $23.80$ & $<10^{-15}$ & \checkmark \\
total\_rooms       & $7.50\!\times\!10^{8}$ & $5.82\!\times\!10^{7}$ & $12.89$ & $<10^{-15}$ & \checkmark \\
population         & $5.90\!\times\!10^{8}$ & $4.31\!\times\!10^{7}$ & $13.67$ & $<10^{-15}$ & \checkmark \\
housing\_age       & $4.08\!\times\!10^{8}$ & $4.29\!\times\!10^{7}$ & $9.51$  & $<10^{-15}$ & \checkmark \\
total\_bedrooms    & $3.42\!\times\!10^{8}$ & $2.89\!\times\!10^{7}$ & $11.84$ & $<10^{-15}$ & \checkmark \\
households         & $2.60\!\times\!10^{8}$ & $2.59\!\times\!10^{7}$ & $10.02$ & $<10^{-15}$ & \checkmark \\
\bottomrule
\end{tabular}